\documentclass[10pt,twocolumn,twoside]{IEEEtran}
\usepackage[dvips]{graphicx}
\usepackage[english]{babel}
\usepackage[latin1]{inputenc}
\usepackage{amssymb}
\usepackage{pmat}
\usepackage{balance}
\usepackage{subfigure}
\usepackage{multirow}
\usepackage{longtable}
\usepackage{amsthm}
\usepackage{amssymb }
\usepackage{amsmath}
\usepackage{bm}
\usepackage{mathrsfs}
\usepackage{amsfonts}
\usepackage{longtable}
\usepackage{multirow}
\usepackage{algorithm, algpseudocode}
\usepackage{cite}
\usepackage{stfloats}

\def\Nobs{{\aleph^{(l)}_\text{obs}}}
\def\Nfuse{\aleph^{(l)}_\text{fuse}}

\def\x{\bm{x}}

\def\cx{{{\mathbb{X}}}}
\def\CF{\text{FF}}
\def\LF{\text{LF}}

\newtheorem{rst}{Result}

\newtheorem{lemma}{\bf{Lemma}}

\graphicspath{{New_Fig/}}

\algnewcommand\Input{\item[\hspace{6pt}\textbf{Input:}]}
\algnewcommand\Output{\item[\hspace{6pt}\textbf{Output:}]}
\algnewcommand\OutputVal{\textbf{output} }

\def\FOC{\bm{C}}

\def\FOA{\bm{A}}
\def\FOE{\bm{E}}

\def\NL{N_L}
\def\NS{N_s}
\def\z{\bm{z}}

\def\y{\bm{y}}

\def\NF{N_f}

\def\Lqo{Y^{(l,m)}}
\def\Lqok{Y^{(l,m)}(k)}

\def\Lok{Z^{(l,m)}(k)}
\def\Lrk{R^{(l,m)}(k)}
\def\Lgk{\bm{g}^{(l,m)}(\x(k))}
\def\Lg{\bm{g}^{(l,m)}}

\def\LLk{\bm{L}^{(l)}(k+1)}
\def\LLpk{\bm{L}^{(l)}(k+1|k)}

\def\LL{\bm{L}^{(l)}}

\def\LLQ{\LL_{\text{Q}}}
\def\LLQk{\LL_{\text{Q}}(k+1)}
\def\Lhk{\bm{h}_i^{(l,m)}(k)}

\def\GLk{\bm{L}^{(\text{G})}(k+1)}

\def\GQLk{\bm{L}_{\text{Q}}^{(\text{G})}(k\!+\!1)}
\def\GQL{\bm{L}_{\text{Q}}^{(\text{G})}}

\def\QC{\bm{C}_{Q}}
\def\JAUXlk{\bm{J}^{(l)}_{\text{AUX}}(k)}
\def\Il{\bm{I}^{(l)}}
\def\PCl{P^{(l)}_c(k\!+\!1)}
\def\JAUXlAll{\bm{J}^{(l)}_{\text{AUX}}(0\colon\!k)}

\title{Distributed Computation of the Conditional PCRLB for
  Quantized Decentralized Particle Filters}
\author{
\authorblockN{Arash Mohammadi$^*$, Amir Asif$^*$, Xionghu Zhong$^+$,
  and A. B. Premkumar$^+$}\\
\authorblockA{$^*$Computer Science and Engineering, York University,
  Email: $\{$marash, asif$\}$@cse.yorku.ca.}\\
\authorblockA{$^+$Computer Engineering, Nanyang Technological
  University, Email: $\{$xhzhong, asannamalai$\}$@ntu.edu.sg.}
}
\begin{document}
\maketitle
\begin{abstract}
The conditional posterior Cram\'er-Rao lower bound (PCRLB) is an
effective sensor resource management criteria for large,
geographically distributed sensor networks.  Existing algorithms for
distributed computation of the PCRLB (dPCRLB) are based on raw
observations leading to significant communication overhead to the
estimation mechanism.  This letter derives distributed computational
techniques for determining the conditional dPCRLB for quantized,
decentralized sensor networks (CQ/dPCRLB). Analytical expressions for the
CQ/dPCRLB are derived, which are particularly useful
for particle filter-based estimators. The CQ/dPCRLB is compared for
accuracy with its centralized counterpart through Monte-Carlo
simulations.
\end{abstract}
%
\section{Introduction} \label{sec:Introduction}
%
Recent developments in sensor technologies and advances in
communication systems allow a large number of observation nodes
(sensors) to be deployed in geographically distributed sensor networks
typically configured using the decentralized topology without
employing a global fusion centre.  To elaborate further, the letter
considers a commonly reported decentralized sensor network
topology~\cite{Tharmarasa:2011} with two types of nodes:
(i)~\textit{Sensor nodes}: with limited power used to record
measurements,~and; (ii)~\textit{Local processing nodes}: responsible
for selecting sensors, processing the data locally, and cooperating
distributively with other connected processing nodes to reach a
consensual tracking estimate for the target. In such geographically
distributed sensor networks, limitations in power budget, system
bandwidth, and communication capabilities impose two critical
restrictions.  First, the maximum number of active sensors at a
particular time is constrained.  Second, only quantized observations
are exchanged between the sensors and processing nodes.  Within its
observation neighbourhood, a local processing node, therefore,
activates a small subset of sensors to receive the quantized version
of their observations.

\noindent
\textbf{Motivation}: The Posterior Caram\'er-Rao Lower Bound (PCRLB)
has been used as an effective criteria for adaptive sensor resource
management problems~\cite{Tharmarasa:2011,Zuo:2011} because it
provides a near-optimal bound of the achievable estimator's
performance and can be computed predictively. Further, it is
independent of and not constrained by a specific estimation
methodology.  Existing PCRLB
derivations~\cite{Msechu:2008,Zhou:2010,Duan:2008}
using quantized observations are limited to centralized estimation
architectures and have not yet been extended to decentralized
topologies.
The letter addresses this gap and derives the PCRLB for decentralized
state estimation in sensor networks with \textit{quantized}
observations. We refer to the distributed computation of the PCRLB as
dPCRLB.  Our previous work~\cite{Arash:TSP1} derives distributed
expressions for computing the \underline{non-conditional} dPCRLB for
full-order decentralized state estimation.  In~\cite{Arash:SPL}, we
extend our derivations to the \underline{conditional} dPCRLB (where
instead of using statistics for the observation model, actual
observations from the previous iterations are utilized).
Both~\cite{Arash:TSP1} and~\cite{Arash:SPL} consider raw observations
in the dPCRLB derivations, as is also the status quo in the
decentralized estimation literature~\cite{Tharmarasa:2011}.  Such a
setup leads to a large communication overhead between the sensors and
their associated processing node making the system impractical.

\noindent
\textbf{Contributions}: The letter extends the conditional dPCRLB
framework to quantized observations with emphasis on particle filter
estimators. Additional contributions of the letter include:
%
\noindent
(a) Both computational and communication complexity
of~\cite{Arash:SPL} are reduced in the proposed conditional dPCRLB
with quantized observations (CQ/dPCRLB).  (b) In~\cite{Arash:SPL}, the
conditional Fisher information matrix (FIM), i.e., the inverse of the
conditional dPCRLB, is expressed as a function of the auxiliary FIM
which is updated distributively at each iteration.  The CQ/dPCRLB updates the
conditional dPCRLB directly without the need of computing the
auxiliary FIM leading to significant communication savings.
%
%

\section{System Description} \label{sec:background}
We consider the non-linear dynamical system
\begin{eqnarray} \label{eq1:sec2-1}
 \x(k) = \bm{f}(\x(k-1)) + \bm{\xi}(k),
\end{eqnarray}
where the state vector $\x = [ X_1, X_2, \ldots , X_{n_x}]^T$ and
$\bm{\xi}(k)$ is the global uncertainty in the state model at
iteration $k$.
Processing node $l$, ($1 \leq l \leq \NF$), is connected to a set of sensor nodes: a subset of which  is
 active at each iteration. The active sensors connected to node $l$ constitute its local observation
neighbourhood $\Nobs$. The total number of active sensors in the network~is
$\NS = \sum_{l=1}^{\NF} |\Nobs|$, where $|\cdot|$ denotes the
cardinality operator.
Sensor $m$ in the observation neighbourhood of
node $l$, i.e., $m \in \Nobs$, makes observation $\Lok$. Instead of
transferring the raw observation, sensor $m$ communicates its
quantized version $\Lqok$ to the fusion node $l$ based on the
following model
%
\begin{eqnarray}
\Lqok = \bm{Q}^{(l,m)}\big(\underbrace{\Lgk +\bm{\zeta}^{(l,m)}(k)}_{\Lok}\big),
\end{eqnarray}
where $\bm{Q}^{(l,m)}(\cdot)$ is the local quantization operator at
node $l$, and $\Lg(\cdot)$ and $\bm{\zeta}^{(l,m)}(\cdot)$ are,
respectively, the local observation model and uncertainty at sensor
$m$ connected to fusion node $l$.  For simplicity and without loss of
generality, the quantization operators $\bm{Q}^{(l,m)}(\cdot)$ are
considered to be the same across the network (i.e.,
$\bm{Q}^{(l,m)}(\cdot) = \bm{Q}(\cdot)$).  Collectively, the overall
quantized observation vector at node $l$ is denoted~by
\begin{eqnarray}
\y^{(l)}(k) = \{ \Lqok: m \in \Nobs\}, \quad\text{for } (1 \leq l \leq \NF).
\end{eqnarray}
Depending on how many sensors are activated by the processing node
$l$, the dimension of observation vector $\y^{(l)}(k)$ is different at
each processing node.  As for the quantized observations
$\y^{(l)}(k)$, vector $\z^{(l)}(k)$ is the collection of all raw
observations associated with the processing node $l$, i.e.,
\begin{eqnarray}
\z^{(l)}(k) = \{ \Lok: m \in \Nobs\}, \quad\text{for } (1 \leq l \leq \NF).
\end{eqnarray}
%
We consider an $\NL$-bit quantization scheme, where
node $m$'s quantized observation $\Lqok$ can take any discrete value
between $0$ and $2^{\NL}-1$.  The set of quantization threshold is
denoted by $\bm{q}\! =\! [q_0, q_1, \ldots, q_{2^{\NL}-1}]$ where for
brevity $q_0 \!=\! -\infty$ and $q_{2^{\NL}}\! =\! \infty$.
The likelihood that $\Lqok$ is at level $q_i$~is denoted by $\Lhk
\triangleq P(\Lqok=q_i|\x(k))$ with
\begin{eqnarray}\label{hdef}
\lefteqn{\Lhk = P\big(q_i \leq \Lok \leq q_{i+1}|\x(k)\big)}\\
&&\!\!\!\!\!\!\!\!\!\!\!\!=\! P\left(\!\big[q_i\!-\!\Lgk\big]\!\leq\!
\bm{\zeta}^{(l,m)}(k) \!\leq\! \big[q_{i+1}\!-\!\Lgk \big]\!\right).
 \nonumber
\end{eqnarray}
%
Section~\ref{sec:cpcrlb} reviews the local conditional dPCRLB for raw
observations as presented in~\cite{Arash:SPL} with one proposed
modification.
\section{Conditional $d\text{PCRLB}$ for Raw Observations} \label{sec:cpcrlb}
Based on the conditional PCRLB inequality, the mean square error
associated with the local estimate $\hat{\x}^{(l)}(0\colon\!k\!+\!1)$ of
the state vector at node $l$ is lower bounded~as~follows
\begin{eqnarray}
\mathbb{E}_{\PCl} \big\{\bm{e}^{(l)}(0\!:\!k\!+\!1) (\bm{e}^{(l)}(0\!:\!k\!+\!1))^T\big\} \geq\! [\bm{I}^{(l)}(0\!:\!k\!+\!1)]^{-1},\nonumber
\end{eqnarray}
where $\PCl \!\!\triangleq\!\! P(\x(0\colon\!k),
\z^{(l)}(k\!+\!1)|\z^{(l)}(1\colon\!k))$, $\mathbb{E}\{\cdot\}$ denotes
expectation, and $\bm{e}^{(l)}(0\colon\!k\!+\!1) \triangleq \x(0\colon \!k\!+\!1)-\hat{\x}^{(l)}(0\colon \!k\!+\!1) $ is the estimation
error.
Defining the 1st and 2nd order partial derivatives as
$\nabla_{\x(k)} = \big[\frac{\partial}{\partial X_1(k)}, \ldots,
  \frac{\partial}{\partial X_{n_x}(k)}\big]^T$ and
$\Delta^{\x(k)}_{\x(k-1)} = \nabla_{\x(k-1)}\nabla_{\x(k)}^T$,
%
the local \underline{accumulated} conditional FIM $\Il(0\!:\!k\!+\!1)$
corresponds to the state trajectory $\hat{\x}^{(l)}(0\colon \!k\!+\!1)$ from iteration $0$ to $k\!+\!1$ and is given~by
\begin{equation} \label{eq/def/FIM}
\bm{I}^{(l)}(0\!:\!k\!+\!1) \triangleq
\mathbb{E}_{P^{(l)}_c(k+1)}\big\{-\Delta^{\x(0:k+1)}_{\x(0:k+1)} \log P^{(l)}_c(k\!+\!1)\big\}.
\end{equation}
Another local FIM  is the local \underline{instantaneous} conditional FIM $\LLk$ associated with $\hat{\x}^{(l)}(k\!+\!1)$, which is obtained by taking
the inverse of $(n_x\times n_x$) right-lower block of $[\bm{I}^{(l)}(0\colon\!k\!+\!1)]^{-1}$.  Please refer to
 Appendix~A for differences in the two FIMs.
Node $l$ updates $\LLk$ as follows.
\begin{rst}\label{Varsh}
The instantaneous local FIM $\bm{L}^{(l)}(k+1)$ associated with estimate
$\hat{\x}^{(l)}(k\!+\!1)$ at node $l$ is computed as follows
\begin{eqnarray} \label{ext.crlb.1}\label{qqww1}
\!\!\!\!\!\!\!\!\!\!\!\!\!\!\!\!\!\!&&\bm{L}^{(l)}(k+1)  \approx \big[\bm{B}^{22}(k)\big]^{(l)}\\
\!\!\!\!\!\!\!\!\!\!\!\!\!\!\!\!\!\!&&~~~ -\big[\bm{B}^{21}(k)\big]^{(l)}\Big( \bm{L}^{(l)}(k)\!+\!
\big[\bm{B}^{11}(k)\big]^{(l)}\Big)^{-1}\big[\bm{B}^{12}(k)\big]^{(l)},  \nonumber\\
\!\!\!\!\!\!\!\!\!\!\!\!\!\!\!\!\!\!&&\mbox{where } \big[\bm{B}^{11}(k)\big]^{(l)} \!\!= \mathbb{E} \big\{ \!\!-\!\Delta^{\x(k)}_{\x(k)}
\log P\big(\x(k+1)|\x(k)\big)\big\}, \label{Qcond/d11}\\
\!\!\!\!\!\!\!\!\!\!\!\!\!\!\!\!\!\!&&\big[\bm{B}^{12}(k)\big]^{(l)} \!\!\!=\!  \mathbb{E}\big\{\!\!\!-\!\!\Delta^{\x(k+1)}_{\x(k)}
\log P\big(\x(k+1)|\x(k)\big)\!\big\}\label{Qcond/d12}
\end{eqnarray}
\begin{eqnarray}
\!\!\!\!\!\!\!\!\!\!\!\!\!\!\!\!\!\!&&\mbox{and } \big[\bm{B}^{22}(k)\big]^{(l)} = \mathbb{E} \big\{\!\! -\!\!\Delta^{\x(k+1)}_{\x(k+1)} \log P\big(\x(k+1)|\x(k)\big)\big\}\nonumber\\
\!\!\!\!\!\!\!\!\!\!\!\!\!\!\!\!\!\!&&~~~~~~~~~~+\mathbb{E} \big\{\!\! -\!\!\Delta^{\x(k+1)}_{\x(k+1)}
\log P\big(\z^{(l)}(k\!+\!1)|\x(k\!+\!1)\big)\big\}. \label{qqww3}
\end{eqnarray}
\end{rst}
%
\noindent
The derivation of Result~\ref{Varsh} is included in Appendix~B.
In~\cite{Arash:SPL}, $\bm{L}^{(l)}(k\!+\!1)$  is computed recursively from the local instantaneous auxiliary FIM $\JAUXlk$ which is the inverse of $(n_x \times n_x)$ right-lower square block of~the accumulated auxiliary FIM $[\JAUXlAll]^{-1}$. The latter is defined~as
\begin{eqnarray}\label{eq/def/Aux/FIM}
\JAUXlAll \triangleq \mathbb{E}_{P^{(l)}_a(k\!+\!1)}
\big\{-\Delta^{\x(0:k)}_{\x(0:k)} \log P^{(l)}_a(k)\big\}
\end{eqnarray}
with $P^{(l)}_a(k)\! \triangleq\!P(\x(0\colon\!k)|\z^{(l)}(1\colon\!k))$.
The  algorithm proposed in~\cite{Arash:SPL}, therefore, requires decentralized
 fusion  of both the local FIMs and the local auxiliary FIMs, while
Result~\ref{Varsh}
eliminates the need for fusing the local instantaneous auxiliary FIMs and, therefore,  cuts the communication overhead by half.

Distributed computation of the conditional PCRLB requires a recursive expression for the predictive local conditional FIM $\bm{L}^{(l)}(k+1|k)$ which is~similar to~\eqref{qqww1} except $[\bm{B}^{22}(k)]^{(l)}$ is
substituted with $[\bm{B}_p^{22}(k)]^{(l)}$ as follows
\begin{equation} \label{ext.crlb.1}
\big[\bm{B}_p^{22}(k)\big]^{(l)}\!\! =\! \mathbb{E} \big\{\!\! -\!\!\Delta^{\x(k+1)}_{\x(k+1)}
\log P\big(\x(k\!+\!1)|\x(k)\big)\big\}.
\end{equation}
 Having computed the  local FIMs $\LLk$ and the local prediction FIMs
 $\LLpk$ at iteration $k+1$, the next step in the conditional dPCRLB
 is to fuse these local FIMs to compute the global instantaneous conditional FIM
 $\GLk$.  Reference~\cite{Arash:SPL} derives a fusion rule for
 assimilating local conditional FIMs into the global conditional FIM
 when raw observations are available at each local node.
Section~\ref{QdPCRLB} extends our derivations
 to quantized observations and eliminates the need for fusion of local instantaneous auxiliary FIMs.
\section{$\text{CQ}/d\text{PCRLB}$ with Quantized Observations} \label{QdPCRLB}
In Proposition~\ref{Varsh}, raw observations $\Lok$ are replaced with
their quantized version $\Lqok$, which results in the quantized
filtering conditional FIM $\LLQ(k\!+\!1)$.  Since terms
$[\bm{B}^{11}(k)]^{(l)}$, $[\bm{B}^{12}(k)]^{(l)}$,
$[\bm{B}^{21}(k)]^{(l)}$ are based on the state model, they remain the
same. Term $[\bm{B}^{22}(k)]^{(l)}$ in Eq.~\eqref{qqww3} is now
computed using the quantized observation as follows
\begin{eqnarray}\label{d22/q}
\!\!\!\!\!\!\!\!\!\!\!\!\!\!\!\!\!\!&&\big[\bm{B}_{\text{Q}}^{22}(k)\big]^{(l)} = \mathbb{E} \big\{\!\! -\!\!\Delta^{\x(k+1)}_{\x(k+1)} \log P\big(\x(k+1)|\x(k)\big)\big\}\\
\!\!\!\!\!\!\!\!\!\!\!\!\!\!\!\!\!\!&&~~~~~~~~~~+\underbrace{\mathbb{E} \big\{\!\! -\!\!\Delta^{\x(k+1)}_{\x(k+1)}
\log P\big(\y^{(l)}(k\!+\!1)|\x(k\!+\!1)\big)}_{\bm{J}(\y^{(l)}(k+1))}\!\!\big\}. \nonumber
\end{eqnarray}
To compute $\bm{J}(\y^{(l)}(k\!+\!1))$, likelihood
$P(\y^{(l)}(k\!+\!1)|\x(k\!+\!1))$ along with the second derivative of
its logarithmic function is needed. Because of quantized observations,
$P(\y^{(l)}(k+1)|\x(k+1))$ transforms into a probability mass
function that is discrete with second derivative replaced by a double
summation as described below.
%
Given the state variables, local observations are assumed independent
such that
\begin{eqnarray}\label{eq:Jyl}
\lefteqn{\bm{J}(\y^{(l)}(k+1)) =
\sum_{m \in \Nobs(k)} \bm{J}\big(Y^{(l,m)}(k+1)\big),}\\
\lefteqn{\text{where }\quad \bm{J}(Y^{(l,m)}(k\!+\!1)) }\label{eq:sdLikelihood}\\
&&=\sum_{i=1}^{\NL}-\mathbb{E} \Big\{\delta(\Lqo(k+1)-i)\Delta^{\x(k)}_{\x(k)}
\log(\Lhk)\Big\}\nonumber
\end{eqnarray}
and $\delta(\cdot)$ is the delta function. We note that
$\mathbb{E}\{\delta(\Lqo(k+1)-i)\} \!=\! \Lhk$, where $\Lhk$ was
defined immediately after Eq.~\eqref{hdef} previously and has the
second derivative
\begin{equation} \label{eq:g}
\Delta^{\x(k)}_{\x(k)}
\log(\!\Lhk)\!\!=\!\!\left[\!\!\!\!
\begin{array}{ccc}
\!\frac{\partial^2 \log(\Lhk)}{(\partial(X_1(k)))^2} &\!\!\!\!\!\!\!\!\! \ldots \!\!\!\!&\!\!\!\!\frac{\partial^2 \log(\Lhk)}{\partial(X_1(k))\partial(X_{n_x}(k))}\\
\vdots&\!\!\!\!\!\!\!\!\ddots \!\!\!\!&\!\!\!\!\vdots\\
\frac{\partial^2 \log(\Lhk)}{\partial(X_{n_x}(k))\partial(X_1(k))} &\!\!\!\!\!\! \ldots \!\!\!\!&\!\!\!\!\frac{\partial^2 \log(\Lhk)}{(\partial(X_{n_x}(k)))^2}\\
\end{array}
\!\!\!\!\right]\!\!.
\end{equation}
Under mild regularity conditions, the expected value of
\eqref{eq:g} is equal to the variance of its first moment,~i.e.,
\begin{eqnarray}\label{eq:likPartition}
\mathbb{E} \Big\{\frac{\partial^2 \log\big(\Lhk\big)}{\partial(X_{j}(k))\partial(X_{u}(k))}\Big\}  \!\!=\!\! -
\mathbb{E} \Big\{\frac{\frac{\partial \Lhk}{\partial(X_{j}(k))}\frac{\partial \Lhk}{\partial(X_{u}(k))}}{\big(\Lhk\big)^2}\Big\}.
\end{eqnarray}
Eqs.~\eqref{eq:Jyl}-\eqref{eq:likPartition} are used to compute
$[\bm{B}_{\text{Q}}^{22}(k)]^{(l)}$.  Finally, the local quantized
filtering FIM is given by
\begin{eqnarray} \label{qlFIM}
\lefteqn{\LLQ\big(k+1\big)  \approx \big[\bm{B}_{\text{Q}}^{22}(k)\big]^{(l)}}\\
&&\!\!\!\! -\big[\bm{B}^{21}(k)\big]^{(l)}\Big( \LLQ(k)\!\!+\!\!
\big[\bm{B}^{11}(k)\big]^{(l)}\Big)^{-1}\big[\bm{B}^{12}(k)\big]^{(l)}.  \nonumber
\end{eqnarray}
Eq.~\eqref{qlFIM} is derived  by applying the following factorization
\begin{eqnarray}\label{eq:localLQ}
\lefteqn{\!\!\!\!\!\!\!\!\!\!\!\!P(\x(0:k+1), \y^{(l)}(1:k+1)) = P\big(\x(0:k), \y^{(l)}(1:k)\big)\nonumber}\\
&&\times P\big(\x(k+1)| \x(k)\big) P(\y^{(l)}(k\!+\!1)| \x(k\!+\!1))
\end{eqnarray}
to the quantized version of Eq.~\eqref{eq/def/FIM} and then taking the
inverse of the ($n_x\times n_x$) right lower block of
$[\bm{I}_{\text{Q}}^{(l)}(0\!:\!k\!+\!1)]^{-1}$.  The similarity
between Eqs.~\eqref{qqww1} and~\eqref{qlFIM} is intuitively pleasing.
The local predictive FIM $\LLQ(k\!+\!1|k)$ is derived in the similar
manner as~\eqref{qlFIM} with $[\bm{B}^{22}(k)]^{(l)}$ replaced
by~\eqref{ext.crlb.1} 

\vspace{.05in}
\noindent
\textbf{Fusing Local FIMs (CQ/dPCRLB)}: Result~\ref{qdPCRLB} provides
a fusion rule for assimilating the local FIMs with quantized
observations to compute the global quantized FIM.
%
\begin{rst}\label{qdPCRLB}~The sequence $\{\bm{L}_{\text{Q}}^{(\text{G})}(k+1)\}$
corresponding to the global information submatrix (CQ/dPCRLB) with
quantized local observations follows the following recursion
\begin{eqnarray}
&&\!\!\!\!\!\!\!\!\!\!\!\!\!\!\!\!\!\!\!\!\!\!\!\!\!\GQLk \! \approx\!  \QC^{22}(k)\!\!-\!
\QC^{21}\!(k)\big( \GQL\!(k)\!+\!\QC^{11}\!(k)\big)^{\!-1}\!\QC^{12}\!(k)\!\! \\
&&\!\!\!\!\!\!\!\!\!\!\!\!\!\!\!\!\!\!\!\!\!\!\!\!\!\text{where } \QC^{11}(k) = \mathbb{E}
 \big\{ -\Delta^{\x(k)}_{\x(k)} \log P\big(\x(k+1)|\x(k)\big)\big\}, \label{cond/d11}\label{cond/QFIM}\\
&&\!\!\!\!\!\!\!\!\!\!\!\!\!\!\!\!\!\!\!\!\!\!\!\!\!\QC^{12}(k) =\mathbb{E}
\big\{\!\!-\Delta^{\x(k+1)}_{\x(k)} \log P\big(\x(k+1)|\x(k)\big)\!\big\},\label{cond/Qd12}\\
&&\!\!\!\!\!\!\!\!\!\!\!\!\!\!\!\!\!\!\!\!\!
\text{and }\bm{C}_{\text{Q}}^{22}(k) \!\approx\! \sum_{l=1}^{N_f} \LLQk-\sum_{l=1}^{N_f} \LLQ(k+1|k) \label{eq:c22/q}\\
&&~~~~\!
+ \mathbb{E} \big\{ \!\!-\!\!\Delta^{\x(k+1)}_{\x(k+1)} \log P\big(\x(k+1)|\x(k)\big)\big\}.  \nonumber
\end{eqnarray}
\end{rst}
\noindent
The proof of Result~\ref{qdPCRLB} is included in Appendix~C.

\vspace{.05in}
\noindent
\textbf{Gaussian Observation Noise}: We derive analytical expressions
for the case when local observations $\Lok$ are zero-mean Gaussian
with variance $\Lrk$, i.e., $\Lok\sim \mathcal{N}(0,\Lrk)$.  The
likelihood that $\Lqok$ is at level $q_i$~is
\begin{eqnarray}\label{eq:hi}
&&\!\!\!\!\!\!\!\!\!\!\!\!\Lhk \!\!=\!\! \frac{1}{\sqrt{2\pi \Lrk}}\!\!\int_{q_i -\Lgk}^{q_{i\!+\!1}  -\Lgk} \!\!\!\!\!\!\!\exp\big\{ \frac{-t}{2\Lrk}\big\}dt\nonumber\\
&&\!\!\!\!\!\!\!\!\!\!\!= \Phi\left( \frac{q_{i}-\Lgk}{\sqrt{\Lrk}}\right) - \Phi\left(\frac{q_{i+1}-\Lgk}{\sqrt{ \Lrk}} \right)\!,
\end{eqnarray}
where $\Phi(\cdot)$ is the standard cumulative Gaussian distribution.  Based on~\eqref{eq:hi}, each derivative term in~\eqref{eq:likPartition} is represented~as
\begin{eqnarray}\label{eq:ev2}
\lefteqn{\frac{\partial \Lhk}{\partial(X_{u}(k))} = - \frac{\frac{\partial\Lgk}{\partial\x(k)}}{\sqrt{2\pi \Lrk}}\times }\\
&&\!\!\!\!\!\!\!\!\!\!\!\!\!\!\left(\! \exp\Big(\!\frac{\!-(q_{i+1}\!\!-\!\Lgk)^2}{2\Lrk}\!\Big) \!\!-\! \exp\!\Big(\!\frac{-(q_{i}\!\!-\!\Lgk)^2}{2\Lrk}\!\Big)\!\!\right)\!\!.\nonumber
\end{eqnarray}
\noindent
\textit{A. Computation of The Conditional dPCRLB}
%
\begin{figure*}[th]
\centering
\mbox{\subfigure[]{\includegraphics[scale=0.4]{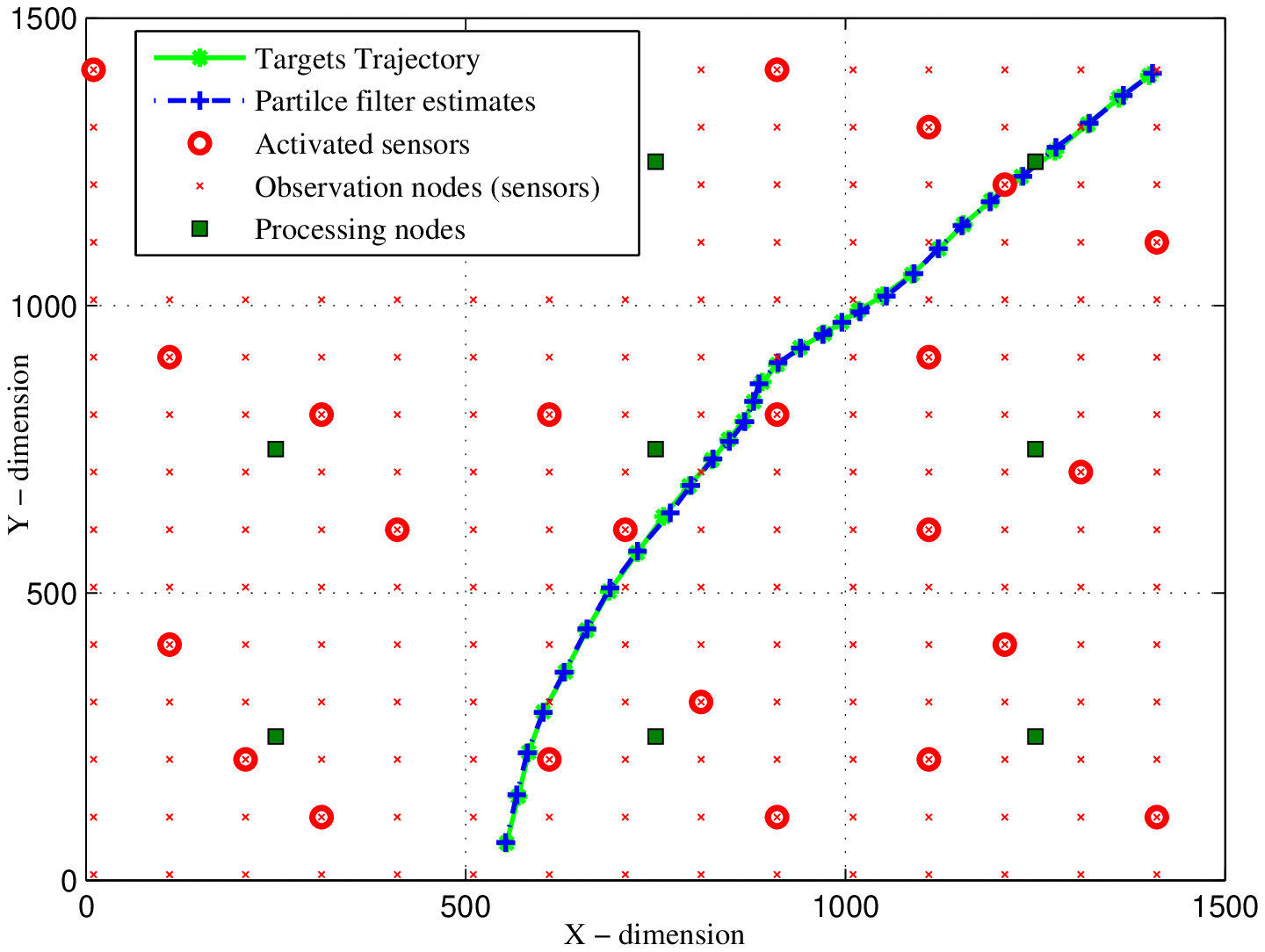}}
\subfigure[]{\includegraphics[scale=0.4]{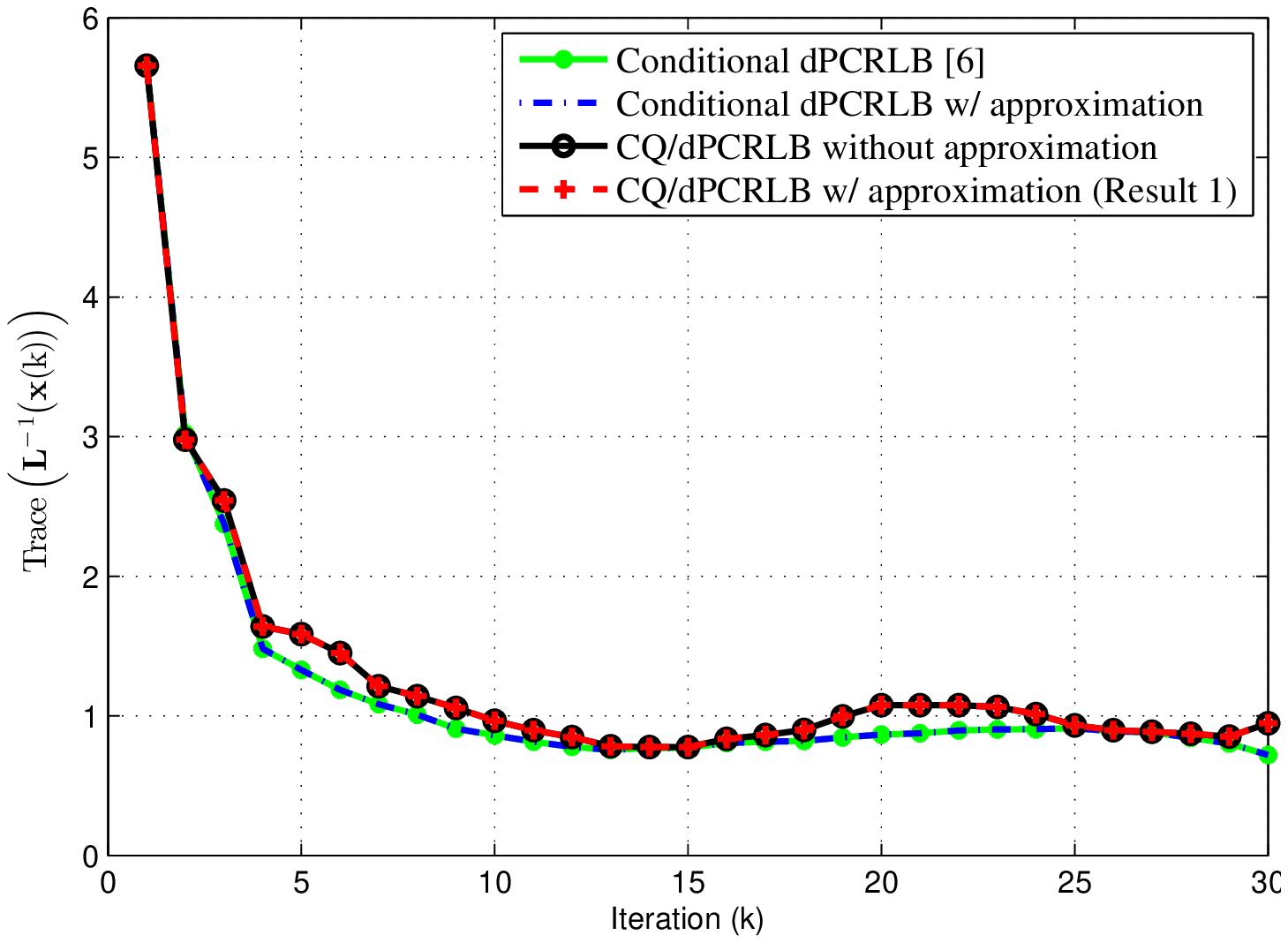}}
\subfigure[]{\includegraphics[scale=0.4]{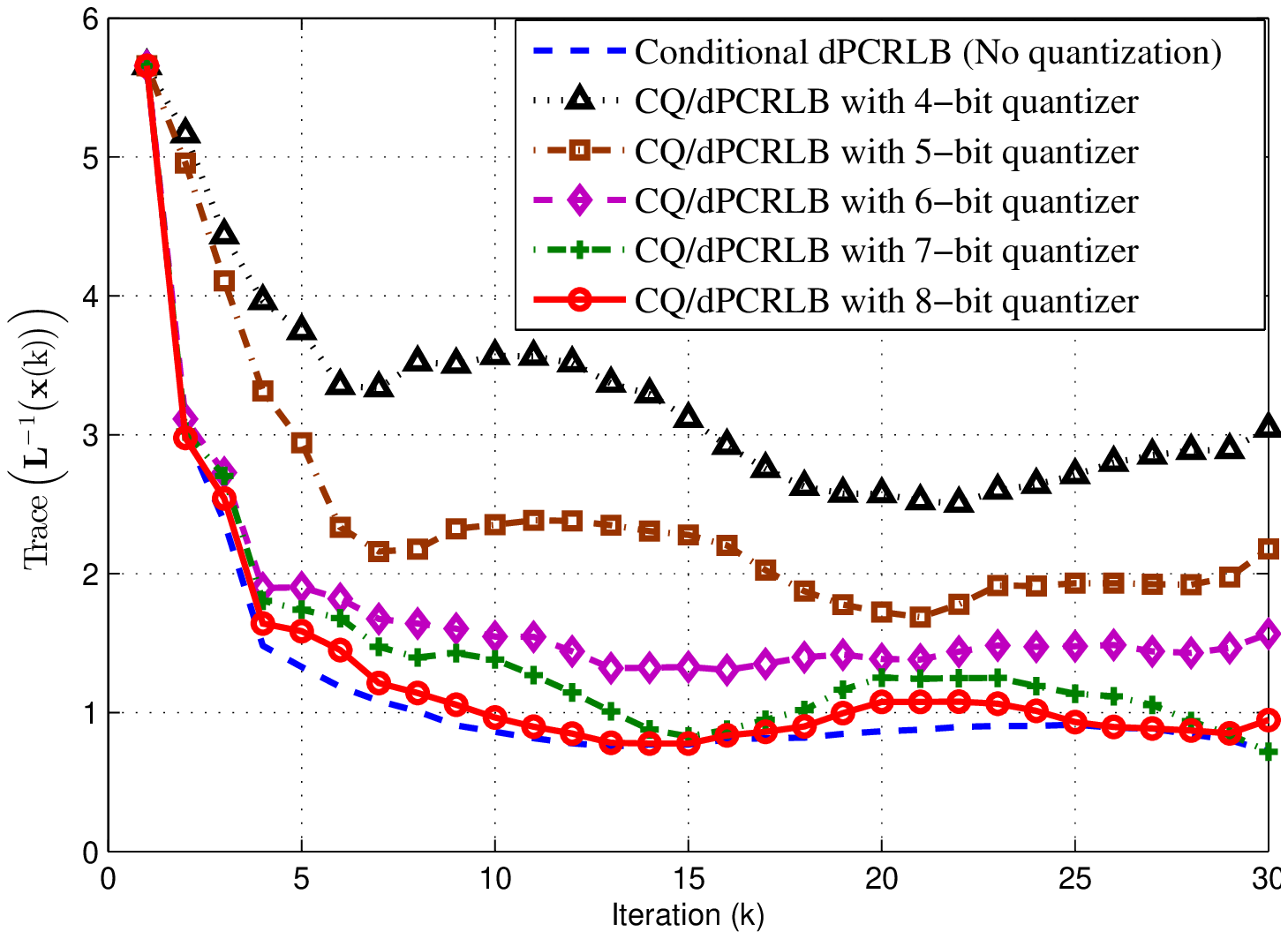}}}
\caption{\label{fig1} $(a)$ A sample decentralized bearing only
  tracking setup. $(b)$ Comparison of the conditional
  dPCRLBs~\cite{Arash:SPL} using raw observations with the CQ/dPCRLBs
  using 8-bit quantized observations. (c) Effect of quantization on
  the CQ/dPCRLB for different (4, 5, 6, 7, and 8 bit) quantization
  levels.}
\vspace{-.1in}
\end{figure*}
%

The analytical computation of the expectations in
Result~\ref{qdPCRLB}
is not practical and,
therefore, particle filter-based approaches are proposed.  If the
state estimator is based on distributed particle
filters~\cite{Arash:TSP1}, then the same particle set can be used in
the CQ/dPCRLB algorithm.
An active sensor  communicates its quantized observation to the associated processing  node. The processing nodes themselves communicate  the local conditional FIMs  and statistics of local posteriors  (i.e., local state estimates and their corresponding covariance matrices) to the neighbouring processing nodes which are then fused in a distributed fashion to compute the global state estimate and the global conditional FIM.
We explain the CQ/dPCRLB algorithm in the
context of the consensus based distributed particle filter
(CF/DPF)~\cite{Arash:TSP1} being used as the state estimator.  The
CF/DPF implements two particle filters at each node:
(i)~Local filter which approximates the local posterior at node $l$
with a set of weighted particles $\{\cx_i^{(l,\LF)}(k),
W_i^{(l,\LF)}\}$, and;
(ii)~Fusion filter which combines the local posteriors to estimate the
global posterior with a second set of particles $\{\cx_i^{(l,\CF)}(k),
W_i^{(l,\CF)}\}$.
All information regarding the observations collected up to time $k$ at
node $l$, are presented in the local particles $\cx_i^{(l,\LF)}(k)$,
while the information available across the network is provided by the
global particles $\cx_i^{(l,\CF)}(k)$.
The CQ/dPCRLB comprises of the following steps:

\noindent
\textit{I. Local FIMs}:
\begin{enumerate}
\item[1.] Eqs.~\eqref{Qcond/d11}-\eqref{Qcond/d12} are computed at
  node $l$ based on Monte-Carlo integration using local particles
  $\cx_i^{(l,\LF)}(k)$.
\item[2.] For computing Eq.~\eqref{d22/q}, first, node $l$ computes
  the predictive particles $\cx_i^{(l,\LF)}(k\!+\!1|k)$ by propagating
  $\cx_i^{(l,\LF)}(k)$ through $P(\x(k\!+\!1)|\x(k))$, and then
  computes Eq.~\eqref{d22/q} using $\cx_i^{(l,\LF)}(k)$ and
  $\cx_i^{(l,\LF)}(k\!+\!1|k)$.
\item[3.] The local FIMs are then computed using Eq.~\eqref{qlFIM}.
\end{enumerate}
\noindent
\textit{II. Global FIM}:
\begin{enumerate}
\item[4.] The expectations in~\eqref{cond/QFIM}-\eqref{eq:c22/q} are
  computed using the global particles $\cx_i^{(l,\CF)}(k)$ to derive
  the FIMs $\bm{C}_{\text{Q}}^{**}(k)$. Eq.~\eqref{eq:c22/q} includes
  summation of local FIMs across the network typically computed using
  the average consensus algorithms~\cite{Arash:SPL} in a decentralized
  fashion.
\item[5.]
Result~\ref{qdPCRLB}
is used to compute the global FIM based on the local FIMs computed in
  Step~4.
\end{enumerate}

\vspace{.05in}
\noindent
\textit{B. Communication Savings}

First, the transfer of quantized observation (instead of raw data)
between sensors and associated processing nodes leads to significant
communication savings.  Second, the communication overhead for
computing the global auxiliary FIM from the local auxiliary FIMs
across the network is eliminated in the proposed CQ/dPCRLB algorithm.
With average consensus~\cite{Arash:TSP1}, the second savings is of
$\text{O}(n_x|\Nfuse|N_c)$ (i.e., the communication complexity reduces by half), where $n_x$ is number of states,
$|\Nfuse|$ the number of processing nodes in the neighbourhood of
processing node $l$, and $N_c$ the number of consensus iterations.
For complexity analysis of the CF/DPF refer to~\cite{Arash:TSP1}.
The CQ/dPCRLB can  be further extended to communicate quantized versions of the local state statistics (quantized local tracks~\cite{Ruan:2005}) and  local FIMs between neighbouring processing nodes during the fusion filter stage which will be considered in future work.

\section{Simulation} \label{sec:simu}
A decentralized bearing-only tracker with nonlinear clockwise
coordinate turn state model~\cite{Arash:TSP1} and observation model
\begin{eqnarray}
  \z^{(l,m)}(k) = {\text{atan}
\left[\frac{X(k)-X^{(l,m)}}{Y(k)-Y^{(l,m)}}\right]}
+ \bm{\zeta}^{(l,m)}(k),
\end{eqnarray}
 is considered. $(X^{(l,m)}, Y^{(l,m)})$ represents the coordinates of
 sensor ($l,m$). Both process and observation noises are normally
 distributed with the observation noise $(\bm{\zeta}^{(l,m)}(k))$
 model assumed to be state dependent such that the bearing noise
 variance at sensor $(l, m)$ depends on the distance between the
 observer and target.
A sensor network (Fig.~\ref{fig1}(a)) consisting of~$225$ static
sensors and $N_f$ = $9$ processing nodes scattered in a square region
of dimension ($1500 \times 1500$)m$^2$ is implemented. Our goal is to
evaluate the performance of the proposed CQ/dPCRLB, therefore, the
activated sensors are selected at random and limited to three sensors
per processing node. For simplicity, the sensors are distributed
uniformly with the processing node at the centre of its rectangular
($500 \times 500$)m neighbourhood. Each processing node communicates
only with its activated sensors within its ($500 \times 500$)m
neighbourhood and other processing nodes within a radius of $550$m.

The objective of our Monte Carlo simulations is three folds. The first
objective is to validate the effectiveness of the conditional FIM
approximation (i.e., to replace the global auxiliary FIM with the
global conditional FIM) in Result~\ref{qdPCRLB}. Fig.~\ref{fig1}(b)
plots the conditional dPCRLB and CQ/dPCRLB with and without the
proposed global conditional FIM approximation. In each case, results
for both raw (bottom two plots) and quantized (top two plots)
observations are included. Within each set of plots in
Fig.~\ref{fig1}(b), the bounds virtually overlap verifying the
effectiveness of the global conditional FIM approximation. The second
objective is to compare the CQ/dPCRLB with quantized observations for
accuracy against the conditional dPCRLB computed from raw
observations~\cite{Arash:SPL}. Comparing bounds across the two sets of
plots in Fig.~\ref{fig1}(b), we note that the respective plots do not
overlap but are fairly close to each other. Despite using quantized
observations, the CQ/dPCRLB is a reasonable approximation of the
dPCRLB. Illustrated in Fig.~\ref{fig1}(c), the third objective is to
quantify the potential CQ/dPCRLB performance loss as a function of the
number of quantization levels. The CQ/dPCRLB approaches the dPCRLB as
the number of quantization levels are increased.

\section{Conclusion} \label{sec:conclusion}
The PCRLB has recently been proposed~\cite{Tharmarasa:2011} as an
effective selection criteria for decentralized sensor resource
management in large, geographically distributed sensor networks.
Existing decentralized algorithms for computing the PCRLB are
typically based on raw observations resulting in a significant
communication overhead. The letter derives the PCRLB for decentralized
estimators in sensor networks with quantized observations and tests it
in a bearing only tracking application.
\appendices
\section{} \label{app:0}
\noindent
Below, we highlight the relationship between the local accumulated conditional FIM $\bm{I}^{(l)}(0\colon \!k\!+\!1)$ and
local instantaneous conditional FIM $\bm{L}^{(l)}(k\!+\!1)$.
The local instantaneous conditional FIM $\bm{L}^{(l)}(k\!+\!1)$ is computed using either of the following three approaches: (i) Directly by inverting large matrix $\bm{I}^{(l)}(0\colon \!k+1)$;
 (ii)  Recursively as a function of the previous local instantaneous auxiliary FIM  $\bm{J}_{\text{AUX}}^{(l)}(k)$~\cite{Arash:SPL}, and; (iii) Recursively as a function of the previous local instantaneous conditional FIM $\bm{L}^{(l)}(k)$ presented in Result~\ref{Varsh}.
In approach (i), first the local accumulated conditional FIM $\bm{I}^{(l)}(0\colon \!k\!+\!1)$ is factorized as follows
\begin{eqnarray}
\bm{I}^{(l)}(0\colon \!k\!+\!1) &=&
 \left[\!\!
\begin{array}{cc}
[A^{11}(k\!+\!1)]^{(l)} & \!\![A^{12}(k\!+\!1)]^{(l)} \\
\left[A^{21}(k\!+\!1)\right]^{(l)} & \!\![A^{22}(k\!+\!1)]^{(l)} \\
\end{array} \!\!\right]\\
&=&  \mathbb{E}
\!\Bigg\{\!\!\!-\!\!\!
\begin{pmat}[{|}]
\!\!\Delta^{\x(0:k)}_{\x(0:k)}  & \Delta^{\x(k+1)}_{\x(0:k)}\!\! \cr\-
\!\!\Delta^{\x(0:k)}_{\x(k+1)}  & \Delta^{\x(k+1)}_{\x(k+1)}\!\! \cr
\end{pmat}
\!\!\log P_c^{(l)}(k\!+\!1)) \!\Bigg\}.\!\!\!\!\nonumber
\end{eqnarray}
Then, the local instantaneous conditional FIM $\bm{L}^{(l)}(k\!+\!1)$ associated with the estimate
$\hat{\x}(k\!+\!1)$  is obtained by taking the inverse of the $(n_x \times n_x)$
right-lower square block of $[\bm{I}^{(l)}(0\colon \!k\!+\!1)]^{-1}$ by applying the following matrix inversion Lemma~\cite{Zuo:2011}.
\begin{lemma}\label{MIL}
Matrix inversion Lemma:
\begin{eqnarray}
\left[
\begin{array}{cc}
\bm{A} & \bm{B}  \\
\bm{B}^T & \bm{C} \\
\end{array} \right]^{-1} =
\left[
\begin{array}{cc}
\bm{\Omega}^{-1} & -\bm{A}^{-1}\bm{B}\bm{\Phi}^{-1}  \\
-\bm{\Phi}^{-1}\bm{B}^{T}\bm{A}^{-1} & \bm{\Phi}^{-1} \\
\end{array} \right],
\end{eqnarray}
where subblocks $\{\bm{A}, \bm{B}, \bm{C}\}$ have conformable
dimensions, $\bm{\Omega} = \bm{A}-\bm{B}\bm{C}^{-1}\bm{B}^{T}$, and
$\bm{\Phi} = \bm{C}-\bm{B}^{T}\bm{A}^{-1}\bm{B}$.
\end{lemma}
\noindent
Based on Lemma~\ref{MIL},  the local instantaneous conditional FIM is given by
\begin{eqnarray}
\bm{L}^{(l)}(k\!+\!1)&=& [\bm{A}^{22}(k\!+\!1)]^{(l)}\\
&-&[\bm{A}^{21}(k\!+\!1)]^{(l)}[\bm{A}^{22}(k\!+\!1)]^{(l)^{-1}}[\bm{A}^{12}(k\!+\!1)]^{(l)}.\nonumber
\end{eqnarray}
which requires inversion of large matrix $[A^{11}(k\!+\!1)]^{(l)}$.
The report describes approach (iii) in more details in Section~III-B.
\section{} \label{app:A}
\begin{figure*}[t]
\normalsize
\setcounter{equation}{34}
\begin{eqnarray} \label{AppA:eq3}
\bm{I}^{(l)}(0\colon \!k\!+\!1) =
\begin{pmat}[{||}]
-\mathbb{E}_{P^{(l)}_c{(k+1)}}\Delta^{\x(0:k-1)}_{\x(0:k-1)} \log P^{(l)}_c(k)
& -\mathbb{E}_{P^{(l)}_c{(k+1)}}\Delta^{\x(k)}_{\x(0:k-1)} \log P^{(l)}_c(k)\!
& \bm{0} \cr\-
-\mathbb{E}_{P^{(l)}_c{(k+1)}}\Delta^{\x(0:k-1)}_{\x(k)} \log P^{(l)}_c(k)
& -\mathbb{E}_{P^{(l)}_c{(k+1)}}\Delta^{\x(k)}_{\x(k)} \log P^{(l)}_c(k)+[\bm{B}^{11}(k)]^{(l)}
&\! [\bm{B}^{12}(k)]^{(l)} \cr\-
\bm{0}
& [\bm{B}^{21}(k)]^{(l)}
& [\bm{B}^{22}(k)]^{(l)} \cr
\end{pmat},
\end{eqnarray}
where $P_c^{(l)}(k) = P(\x(0\colon \!k),\z^{(l)}(k)|\z^{(l)}(1\colon \!k\!-\!1))$.

\setcounter{equation}{29}
\hrulefill
\vspace*{4pt}
\end{figure*}
%
Here Result~1 is derived. We also show that under a minor constraint,  the result in~\cite{Arash:SPL} reduces to Result~1, which is equivalent to replacing the local instantaneous auxiliary FIM $\JAUXlk$  by the local instantaneous conditional FIM $\bm{L}^{(l)}(k)$.
The rational for the approximation is included after the proof.
\begin{proof}[Proof of Result~\ref{Varsh}]
The conditional FIM given observations up to and including time $k\!-\!1$ is factorized as follows
\begin{eqnarray} \label{appA:eq1}
\bm{I}^{(l)}(0\colon \!k) &=&
 \left[\!\!
\begin{array}{cc}
[\FOA^{11}(k)]^{(l)} & \!\![\FOA^{12}(k)]^{(l)} \\
\left[\FOA^{21}(k)\right]^{(l)} & \!\![\FOA^{22}(k)]^{(l)} \\
\end{array} \!\!\right]\\
&=&  \mathbb{E}
\!\Bigg\{\!\!\!-\!\!\!
\begin{pmat}[{|}]
\!\!\Delta^{\x(0:k-1)}_{\x(0:k-1)}  & \Delta^{\x(k)}_{\x(0:k-1)}\!\! \cr\-
\!\!\Delta^{\x(0:k-1)}_{\x(k)}  & \Delta^{\x(k)}_{\x(k)}\!\! \cr
\end{pmat}
\!\!\log P_c^{(l)}(k) \!\Bigg\},\!\!\!\!\nonumber
\end{eqnarray}
where $P_c^{(l)}(k) = P(\x(0\colon \!k),\z^{(l)}(k)|\z^{(l)}(1\colon \!k\!-\!1))$.
Term $\bm{L}^{(l)}(k)$ is the inverse of the right lower block of $[\bm{I}^{(l)}(0\colon \!k)]^{-1}$ which is given by  (using the matrix inversion lemma)
\begin{equation}\label{AppA:eq8}
\bm{L}^{(l)}(k) = [\FOA^{11}(k)]^{(l)} -[\FOA^{21}(k)]^{(l)}[\FOA^{11}(k)]^{(l)^{-1}}[\FOA^{12}(k)]^{(l)}.
\end{equation}
For next iteration $k\!+\!1$, we have
\begin{eqnarray} \label{AppA:eq2}
\lefteqn{\bm{I}^{(l)}(0\colon \!k\!+\!1) =}\\
&& \mathbb{E}\Bigg\{\!\!\!\!-\!\!\!
\begin{pmat}[{||}]
\Delta^{\x(0:k-1)}_{\x(0:k-1)}
& \Delta^{\x(k)}_{\x(0:k-1)}\!
& \Delta^{\x(k+1)}_{\x(0:k-1)} \cr\-
\Delta^{\x(0:k-1)}_{\x(k)}
& \Delta^{\x(k)}_{\x(k)}\!
& \Delta^{\x(k+1)}_{\x(k)} \cr\-
\Delta^{\x(0:k-1)}_{\x(k+1)}
& \Delta^{\x(k)}_{\x(k+1)}\!
& \Delta^{\x(k+1)}_{\x(k+1)} \cr
\end{pmat}
\!\! \log P^{(l)}_c(k\!\!+\!\!1) \!\!\Bigg\},\nonumber
\end{eqnarray}
where $P_c^{(l)}(k\!+\!1) = P(\x(0\colon \!k\!+\!1),\z^{(l)}(k\!+\!1)|\z^{(l)}(1\colon \!k))$ which can be factorized as follows
\begin{eqnarray}\label{eq:axr}
\lefteqn{\!\!\!\!\!\!\!\!\!\!\!\!\!P\big(\x(0\colon \!k\!+\!1),\z^{(l)}(k\!+\!1)|\z^{(l)}(1\colon \!k)\big) \!\!=\!\!P\big(\z^{(l)}(k\!+\!1)|\x(k\!+\!1)\big)\nonumber}\\
&&\!\!\!\!\!\!\!\!\!\!\!\!\!\!\!\!\!\! \times P\big(\x(k\!+\!1)|\x(k)\big) \frac{P\big(\x(0\colon \!k),\z^{(l)}(k)|\z^{(l)}(1\colon \!k\!-\!1)\big)}{P\big(\z^{(l)}(k)|\z^{(l)}(1 \colon \!k\!-\!1)\big) }.
\end{eqnarray}
Taking  logarithm of Eq.~\eqref{eq:axr}.
\begin{eqnarray}
\lefteqn{\!\!\!\!\!\!\!\!\!\!\!\log P_c^{(l)}(k\!+\!1) =\log P\big(\z^{(l)}(k\!+\!1)|\x(k\!+\!1)\big)+ \log P_c^{(l)}(k)\nonumber}\\
&&\!\!\!\!\!\!\!\!\!\!\!\!\!\!\!\!  + \log P\big(\x(k\!+\!1)|\x(k)\big)\!-\! \log P\big(\z^{(l)}(k)|\z^{(l)}(1 \colon \!k\!-\!1)\big).
\end{eqnarray}
Therefore, Eq.~\eqref{AppA:eq2} reduces to Eq.~\eqref{AppA:eq3} on top of the next page
where $[\bm{B}^{11}(k)]^{(l)}$, $[\bm{B}^{12}(k)]^{(l)}$, $[\bm{B}^{21}(k)]^{(l)}$, and $[\bm{B}^{22}(k)]^{(l)}$ are given by Eqs.~\eqref{Qcond/d11}-\eqref{qqww3}.
The four blocks on the top left sub-matrix of Eq.~\eqref{AppA:eq3} are functions of $\z^{(l)}(k)$ which make them different from $[\bm{A}^{**}(k)]^{(l)}$ in Eq.~\eqref{appA:eq1}. In order to recursively compute $\bm{L}^{(l)}(k\!+\!1)$ from $\bm{L}^{(l)}(k)$, these four terms are approximated by their expectations with respect to $P(\z^{(l)}(k)|\z^{(l)}(1\colon \!k\!-\!1))$, i.e.,
\setcounter{equation}{35}
\begin{eqnarray}
\lefteqn{-\mathbb{E}_{P^{(l)}_c{(k+1)}}\Big\{\Delta^{\x(0:k-1)}_{\x(0:k-1)} \log P^{(l)}_c(k)\Big\} \nonumber}\\
\lefteqn{\approx-\mathbb{E}_{P(\z^{(l)}(k)|\z^{(l)}(1\colon \!k\!-\!1))}\Big\{\mathbb{E}_{P^{(l)}_c{(k+1)}}\Delta^{\x(0:k-1)}_{\x(0:k-1)} \log P^{(l)}_c(k)\Big\} \nonumber}\\
\lefteqn{=-\int P(\z^{(l)}(k)|\z^{(l)}(1\colon \!k\!-\!1)) P^{(l)}_c(k\!+\!1)\nonumber}\\
&&\!\!\!\!\!\! \times \Delta^{\x(0:k-1)}_{\x(0:k-1)} \log P^{(l)}_c(k)
 \bm{d}\x(0\colon \!k\!+\!1)\bm{d}\z^{(l)}(k\!+\!1)\bm{d}\z^{(l)}(\!k) \nonumber\\
&&\!\!\!\!\!\!\!\!\!\!\!\!\!\!= [\bm{A}^{11}(k)]^{(l)}.\label{eq:new1}
\end{eqnarray}
Similarly, it can be shown that
\begin{eqnarray}
\!\!\!\!\!\!\!\!-\mathbb{E}_{P^{(l)}_c{(k+1)}}\Big\{\Delta^{\x(k)}_{\x(0:k-1)} \log P^{(l)}_c(k)\Big\} &\!\!\approx\!\!& [\bm{A}^{12}(k)]^{(l)}.\\
\!\!\!\!\!\!\!\!-\mathbb{E}_{P^{(l)}_c{(k+1)}}\Big\{\Delta^{\x(k:k-1)}_{\x(k)} \log P^{(l)}_c(k)\Big\} &\!\!\approx\!\!& [\bm{A}^{21}(k)]^{(l)}.\\
\!\!\!\!\!\!\!\!-\mathbb{E}_{P^{(l)}_c{(k+1)}}\Big\{\Delta^{\x(k)}_{\x(k)} \log P^{(l)}_c(k)\Big\} &\!\!\approx\!\!& [\bm{A}^{22}(k)]^{(l)}.\label{eq:new12}
\end{eqnarray}
Finally, Eq.~\eqref{AppA:eq3} can be approximated as follows
\begin{eqnarray} \label{AppA:eq32}
\lefteqn{\bm{I}^{(l)}(0\colon \!k\!+\!1) \approx }\\
&&
\begin{pmat}[{||}]
 [\bm{A}^{11}(k)]^{(l)}
& [\bm{A}^{12}(k)]^{(l)}\!
& \bm{0} \cr\-
 [\bm{A}^{21}(k)]^{(l)}
&  [\bm{A}^{22}(k)]^{(l)}+[\bm{B}^{11}(k)]^{(l)}
&\! [\bm{B}^{12}(k)]^{(l)} \cr\-
\bm{0}
& [\bm{B}^{21}(k)]^{(l)}
& [\bm{B}^{22}(k)]^{(l)} \cr
\end{pmat}. \nonumber
\end{eqnarray}
Going back to complete the proof, we note
that the information sub-matrix $\bm{L}^{(l)}(k\!+\!1)$ is given
by the inverse of the right bottom ($n_x\times n_x$) block of $[\bm{I}^{(l)}(0\colon \!k)]^{-1}$
(corresponding to $[\bm{B}^{22}(k)]^{(l)}$ in Eq.~\eqref{AppA:eq32}), i.e.,
\begin{eqnarray} \label{AppA:eq5}
\lefteqn{\!\bm{L}^{(l)}(k\!\!+\!\!1) = [\bm{B}^{22}(k)]^{(l)}-\big[\bm{0}~ [\bm{B}^{21}(k)]^{(l)}\big]}\\
\lefteqn{\!\times \left[
\begin{array}{cc}
\!\!\![\bm{A}^{11}(k)]^{(l)} & [\bm{A}^{12}(k)]^{(l)}\!\!\! \\
\!\!\![\bm{A}^{21}(k)]^{(l)} & [\bm{A}^{22}(k)]^{(l)}\!\!+\!\![\bm{B}^{11}(k)]^{(l)}\!\!\! \\
\end{array}
\!\right]^{-1}
\left[
\begin{array}{c}
\!\!\!\bm{0} \!\!\!\\
\!\!\![\bm{B}^{12}(k)]^{(l)} \!\!\!\\
\end{array} \right]\nonumber}\\
\lefteqn{=\! [\bm{B}^{22}(k)]^{(l)}-[\bm{B}^{21}(k)]^{(l)}\Big(\! [\bm{A}^{22}(k)]^{(l)}}\label{AppA:eq6}
\\
&&\!\!\!\!\!\!\!\!\!\!\!\!\! -[\bm{A}^{21}(k)]^{(l)}[\bm{A}^{11}(k)]^{(l)^{-1}}\![\bm{A}^{12}(k)]^{(l)}
\!+\! [\bm{B}^{11}(k)]^{(l)}\!\Big)^{-1}\!\!\!\!\![\bm{B}^{12}(k)]^{(l)}\nonumber
\end{eqnarray}
Based on Eq.~\eqref{AppA:eq8},
the middle term in Eq.~\eqref{AppA:eq6} reduces to $\bm{L}^{(l)}(k)\!+\! [\bm{B}^{11}(k)]^{(l)}$ which by
substituting  in Eq.~\eqref{AppA:eq6} proves Result~1.
\end{proof}

Finally we note that Result 1 is valid with the following approximation:

The top left  four blocks of the accumulated conditional FIM given by Eq.~\eqref{AppA:eq3}
are replaced by their expectations with respect to $P(\z^{(l)}(k)|\z^{(l)}(1\colon \!k\!-\!1))$.

As shown above, this leads to Eqs.~\eqref{qqww1}-\eqref{qqww3} of Result~1. Comparing Eqs.~\eqref{qqww1}-\eqref{qqww3} with our earlier result~\cite{Arash:SPL}, we note that the instantaneous auxiliary FIM $\bm{J}^{(l)}_{\text{AUX}}(k)$ is replaced with the instantaneous conditional FIM $\bm{L}^{(l)}(k)$. Consequently, the CQ/dPCRLB updates the
conditional dPCRLB directly without the need of computing the
auxiliary FIM leading to significant communication savings (by a factor of 2).
\section{} \label{app:B}
\begin{figure*}[t]
\normalsize
\setcounter{equation}{52}
\begin{eqnarray} \label{AppA:eq7}
\bm{I}_{\text{Q}}^{(\text{G})}(0\colon \!k\!+\!1) =
\begin{pmat}[{||}]
-\mathbb{E}_{P_{Q, c}{(k+1)}}\Delta^{\x(0:k-1)}_{\x(0:k-1)} \log P_{Q, c}(k)
& -\mathbb{E}_{P_{Q, c}{(k+1)}}\Delta^{\x(k)}_{\x(0:k-1)} \log P_{Q, c}(k)\!
& \bm{0} \cr\-
-\mathbb{E}_{P_{Q, c}{(k+1)}}\Delta^{\x(0:k-1)}_{\x(k)} \log P_{Q, c}(k)
& -\mathbb{E}_{P_{Q, c}{(k+1)}}\Delta^{\x(k)}_{\x(k)} \log P_{Q, c}(k)+\FOC^{11}(k)
&\! \FOC^{12}(k) \cr\-
\bm{0}
& \FOC^{21}(k)
& \FOC^{22}(k) \cr
\end{pmat},
\end{eqnarray}
where $P_{Q, c}(k\!+\!1)  \triangleq P(\x(0\!:\! k\!+\!1),\y(k\!+\!1)|\y(1\!:\!k))$.

\setcounter{equation}{42}
\hrulefill
\vspace*{4pt}
\end{figure*}
%
Below, Result~\ref{qdPCRLB} is proved.
First, we derive Lemma~\ref{lemma/post/fac/CDpcrlb} which provides a factorization of the global quantized conditional
posterior distribution $P_{Q, c}(k+1)$ at iteration $k+1$ as a function of the local quantized conditional
posterior distribution $P^{(l)}_{Q, c}(k+1)$ at iteration $k+1$ and the global  quantized conditional
posterior distribution $P_{Q, c}(k)$ at iteration $k$.
\begin{lemma}\label{lemma/post/fac/CDpcrlb}
Assuming that the quantized observations conditioned on the state variables are
independent, the global posterior for a  network with $N_f$ processing nodes is
factorized as follows
\begin{eqnarray} \label{eq/post/fac/CDpcrlb}
\lefteqn{P_{Q, c}(k\!+\!1) \triangleq P(\x(0\!:\!k\!+\!1), \y(k\!+\!1)|\y(1\!:\!k)) }\\
&&\!\!\!\!\!\!\propto\frac{\prod_{l=1}^{N_f}P^{(l)}_{Q, c}(k\!+\!1)}
{\prod_{l=1}^{N_f} P\big(\x(k\!+\!1)|\y^{(l)}(1\!:\!k)\big)} P\big(\x(k\!+\!\!1)|\x(k)\big)  P_{Q, c}(k),\nonumber
\end{eqnarray}
where $\qquad P_{Q, c}(k) \triangleq P(\x(0\!:\! k),\y(k)|\y(1\!:\!k\!-\!1))$, \\
and $\quad ~P^{(l)}_{Q, c}(k\!+\!1)  \triangleq P(\x(0\!:\! k\!+\!1),\y^{(l)}(k\!+\!1)|\y^{(l)}(1\!:\!k))$.
\end{lemma}
\begin{proof}[Proof of Lemma~\ref{lemma/post/fac/CDpcrlb}]
Using the Markovian property
\begin{eqnarray} \label{eq/cond/fact1}
P_{Q, c}(k\!+\!1) &=& P\big(\y(k\!+\!1)| \x(k\!+\!1)\big) \\
&\times&P\big(\x(k\!+\!1)| \x(k)\big) P\big(\x(0\!:\!k) |\y(1\!:\!k)\big).\nonumber
\end{eqnarray}
Comparing Eq.~\eqref{eq/post/fac/CDpcrlb} with \eqref{eq/cond/fact1}, we need to prove: (i) $P\big(\y(k\!+\!1)| \x(k\!+\!1)\big) \propto \prod_{l=1}^{N_f}P^{(l)}_{Q, c}(k\!+\!1)/P\big(\x(k\!+\!1)|\y^{(l)}(1\!:\!k))$, and; (ii) $P_{Q, c}(k) \propto P\big(\x(0\!:\!k) |\y(1\!:\!k)\big)$.

Relationship (i):
Given the state variables, the observations are assumed to be independent as is the case in most Bayesian estimators.
Then,  the first term on the right hand side (RHS) of~(\ref{eq/cond/fact1})~is  given by
\begin{eqnarray} \label{eq/cond/fact2}
P\big(\y(k\!+\!1)| \x(k\!+\!1) \big) = \prod_{l=1}^{N_f} P\big(\y^{(l)}(k\!+\!1)| \x(k\!+\!1)\big).
\end{eqnarray}
We also factorize the local conditional distribution at node $l$, for ($1\leq l \leq N_f$), as follows
\begin{eqnarray} \label{eq:46}
\lefteqn{P\big(\x(k\!+\!1),\!\y^{(l)}(k\!+\!1)|\y^{(l)}(1\!:\!k)\big) }\\
&&=P\big(\y^{(l)}(k\!+\!1)|\x(k\!+\!1)\big)P\big(\x(k\!+\!1)|\y^{(l)}(1\!:\!k)\big).\nonumber
\end{eqnarray}
In terms of the local likelihood $P(\y^{(l)}(k+1)|\x(k+1))$, Eq.~\eqref{eq:46} can be expressed as follows
\begin{equation} \label{eq47new}
P\big(\y^{(l)}(k\!+\!1)|\x(k\!+\!1)\big) =\frac{P\big(\x(k\!+\!1),\!\y^{(l)}(k\!+\!1)|\y^{(l)}(1\!:\!k)\big)}{P\big(\x(k\!+\!1)|\y^{(l)}(1\!:\!k)\big)}.
\end{equation}
Substituting Eq.~\eqref{eq47new} in Eq.~\eqref{eq/cond/fact2}, we have
\begin{eqnarray}
P\big(\y(k\!+\!1)| \x(k\!+\!1) \big) =\prod_{l=1}^{N_f}\!\!
 \frac{P\left(\x(k\!+\!1),\!\y^{(l)}(k\!+\!1)|\y^{(l)}(1\!:\!k)\right)}
{P\left(\x(k\!+\!1)|\y^{(l)}(1\!:\!k)\right)},\nonumber
\end{eqnarray}
which proves Relation~(i).

\noindent
Relationship~(ii): Term $ P_{Q, c}(k)$ can be factorized as follows
\begin{eqnarray} \label{eq/cond/fact3}
P_{Q, c}(k) = P\big(\x(0\!:\!k) |\y(1\!:\!k)\big)P\big(\y(k) |\y(1\!:\!k\!-\!1)\big).
\end{eqnarray}
Since $P\big(\y(k) |\y(1\!:\!k\!-\!1)\big)$ is independent of the state
variables, Eq.~\eqref{eq/cond/fact3} can be expressed as follows
\begin{eqnarray}
P_{Q, c}(k) \propto P\big(\x(0\!:\!k) |\y(1\!:\!k)\big),
\end{eqnarray}
which proves Relation~(ii).

\noindent
This completes the proof for Lemma~1.
\end{proof}
\begin{proof}[Proof of Result~\ref{qdPCRLB}]
Given the quantized observations up to and including time $k$, the global accumulated conditional FIM  can be decomposed as follows
\begin{eqnarray} \label{cpp}
\bm{I}_{\text{Q}}^{(\text{G})}(0\colon\!k)  &\!\!=\!\!& \mathbb{E}
\Bigg\{\!\!\!-\!\!\!
\begin{pmat}[{|}]
\!\!\Delta^{\x(0:k-1)}_{\x(0:k-1)}  & \Delta^{\x(k)}_{\x(0:k-1)}\!\! \cr\-
\!\!\Delta^{\x(0:k-1)}_{\x(k)}  & \Delta^{\x(k)}_{\x(k)}\!\! \cr
\end{pmat}
\!\!\log P_{Q, c}(k)  \!\Bigg\}\nonumber\\
&\!\!\triangleq\!\!&
 \left[\!\!
\begin{array}{cc}
\FOE^{11}(k) & \!\!\!\!\FOE^{12}(k)  \\
\FOE^{21}(k) & \!\!\!\!\FOE^{22}(k) \\
\end{array} \!\!\right].
\end{eqnarray}
As stated previously in Appendix~A, the  instantaneous conditional FIM $\bm{L}_{\text{Q}}^{(\text{G})}(k)$ is obtained by taking the inverse of the right lower block of $[\bm{I}_{\text{Q}}^{(\text{G})}(0\colon \!k)]^{-1}$. Using Lemma~\ref{MIL} we get
\begin{equation}\label{AppA:eqn4}
\GQL(k) =\FOE^{11}(k) -\FOE^{21}(k)[\FOE^{11}(k)]^{-1}\FOE^{12}(k).
\end{equation}
For iteration $k+1$, we
decompose $\x(0\!:\!k\!+\!1)=[\x^T(0\!:\!k\!-\!1), \x^T(k), \x^T(k\!+\!1)]^T$.
As for Eq.~\eqref{cpp},
 the global accumulated conditional FIM for iteration $k\!+\!1$ is
then given by
\begin{eqnarray}\label{eq49}
\lefteqn{\bm{I}_{\text{Q}}^{(\text{G})}(0:k+1)}\\
\!\!&\!\!\!\!\!\!=\!\!\!\!\!\!&\!\! \mathbb{E}\Bigg\{\!\!-\!\!\!
\begin{pmat}[{||}]
\Delta^{\x(0:k-1)}_{\x(0:k-1)}
& \Delta^{\x(k)}_{\x(0:k-1)}\!
& \Delta^{\x(k+1)}_{\x(0:k-1)} \cr\-
\Delta^{\x(0:k-1)}_{\x(k)}
& \Delta^{\x(k)}_{\x(k)}\!
& \Delta^{\x(k+1)}_{\x(k)} \cr\-
\Delta^{\x(0:k-1)}_{\x(k+1)}
& \Delta^{\x(k)}_{\x(k+1)}\!
& \Delta^{\x(k+1)}_{\x(k+1)} \cr
\end{pmat}
\!\! \log P_{Q,c}(k+1) \!\Bigg\}.\nonumber
\end{eqnarray}
Using Lemma~\ref{lemma/post/fac/CDpcrlb}, Eq.~\eqref{eq49} reduces to Eq.~\eqref{AppA:eq7} given at the top of the page.
Similar to our discussion in Appendix~B,
the four blocks on the top left sub-matrix of Eq.~\eqref{AppA:eq7} are functions of $\y(k)$, which make them different from $\bm{E}^{**}(k)$ in Eq.~\eqref{cpp}.  In order to recursively compute $\bm{L}^{(\text{G})}_{\text{Q}}(k\!+\!1)$  from $\bm{L}^{(\text{G})}_{\text{Q}}(k)$, these four blocks are approximated by taking their expectations with respect to $P(\y(k)|\y(1\colon \!k\!-\!1))$ resulting in
\setcounter{equation}{53}
\begin{eqnarray} \label{fo/J/exp}
\lefteqn{\bm{I}_{\text{Q}}^{(\text{G})}(0:k+1) \nonumber}\\
\!\!&\!\!\!\!\!\!\approx\!\!\!\!\!\!&\!\!
\left[\!\!
\begin{array}{ccc}
\FOE^{11}(k) & \FOE^{12}(k) & \bm{0} \\
\FOE^{21}(k) & \FOE^{22}(k)+ \QC^{11}(k)& \QC^{12}(k) \\
\bm{0} & \QC^{21}(k) & \QC^{22}(k) \\
\end{array} \right],\label{eq:new45}
\end{eqnarray}
where block $\bm{0}$ denotes a block of all zeros.
Terms $\QC^{11}(k)$, $\QC^{12}(k)$ and $\QC^{21}(k)$ were defined previously
 in Eqs.~\eqref{cond/Qd12}-\eqref{eq:c22/q}.
Next, using Lemma~\ref{lemma/post/fac/CDpcrlb}, term $\QC^{22}(k)\!=\! \mathbb{E}\{\!-\!\Delta^{\x(k+1)}_{\x(k+1)}\log P_{Q,c}(k+1)\}$ in Eq.~\eqref{fo/J/exp} is expressed as
\begin{eqnarray} \label{cond/c22-v2}
\lefteqn{
\QC^{22}(k)
\!\!=\!\!   \mathbb{E}_{P_{Q,c}(k+1)}\Big\{\!\!-\!\!\Delta^{\x(k+1)}_{\x(k+1)}\log\big(P\left(\x(k\!+\!1)|\x(k)\right)\big) \Big\}+\nonumber}\\
&&\!\!\!\!\!\!\!\!\!\!\!\sum_{l=1}^{N_f} \!\mathbb{E}_{P_{Q,c}(k+1)}\!\Big\{\!\!\!-\!\!\Delta^{\x(k+1)}_{\x(k+1)}\!\log\!\big(\!P( \x(k\!\!+\!\!1\!), \y^{(l)}(k\!\!+\!\!1\!))|\y^{(l)}(1\!\!:\!\!k))\big)\!\Big\}\nonumber \\
&&\!\!\!\!\!\!\!\!\!\!\! - \sum_{l=1}^{N_f} \mathbb{E}_{P_{Q,c}(k+1)}\Big\{\!\!-\!\!\Delta^{\x(k+1)}_{\x(k+1)}\log\big( P( \x(k\!+\!1)|\y^{(l)}(1\!\!:\!k))\big)\Big\}\!.
\end{eqnarray}
Finally,  we note that the two summation terms in Eq.~\eqref{cond/c22-v2} are  individual sums of the local instantaneous conditional FIMs at iteration $k\!+\!1$, i.e.,
\begin{eqnarray}
&&\!\!\!\!\!\!\!\!\!\!\!\sum_{l=1}^{N_f} \!\mathbb{E}_{P_{Q,c}(k+1)}\!\Big\{\!\!\!-\!\!\Delta^{\x(k+1)}_{\x(k+1)}\!\log\!\big(\!P( \x(k\!\!+\!\!1\!), \y^{(l)}(k\!\!+\!\!1\!))|\y^{(l)}(1\!\!:\!\!k))\big)\!\Big\}  \nonumber\\
&&\approx\sum_{l=1}^{N_f}\bm{L}_{\text{Q}}^{(l)}(k\!+\!1)\nonumber
\end{eqnarray}
and
\begin{eqnarray}
&&\!\!\!\!\!\!\!\!\!\!\!\sum_{l=1}^{N_f} \mathbb{E}_{P_{Q,c}(k+1)}\Big\{\!\!-\!\!\Delta^{\x(k+1)}_{\x(k+1)}\log\big( P( \x(k\!+\!1)|\y^{(l)}(1\!\!:\!k))\big)\Big\} \nonumber\\
&&\approx \sum_{l=1}^{N_f}\bm{L}_{\text{Q}}^{(l)}(k\!+\!1|k).\nonumber
\end{eqnarray}
Term $\QC^{22}(k)$ in Eq.~\eqref{cond/c22-v2}, therefore, reduces to
\begin{eqnarray}
\bm{C}_{\text{Q}}^{22}(k) &\!\approx\!& \sum_{l=1}^{N_f} \LLQk-\sum_{l=1}^{N_f} \LLQ(k+1|k) \nonumber\\
&+& \mathbb{E} \big\{ \!\!-\!\!\Delta^{\x(k+1)}_{\x(k+1)} \log P\big(\x(k+1)|\x(k)\big)\big\}.  \nonumber
\end{eqnarray}
 The information sub-matrix $\GQL(k+1)$
can then be calculated as the inverse of the right lower ($n_x\times n_x$)
sub-matrix of $[\bm{I}_{\text{Q}}^{(G)}(0\colon\!k+1)]^{-1}$  (Eq.~\eqref{fo/J/exp}) as follows
\begin{eqnarray} \label{ext.crlb.12}
\lefteqn{\!\!\!\GQL(k+1) \approx \QC^{22}(k)-}\\
&&\!\!\!\!\!\!\!\!\!\!\!\!\!\!\!\big[\bm{0} \quad \QC^{21}(k)\big]  \left[
\begin{array}{cc}
\FOE^{11}(k) & \FOE^{12}(k) \\
\FOE^{21}(k) & \FOE^{22}(k)+\QC^{11}(k) \\
\end{array}
\right]^{-1}
\!\!\left[
\begin{array}{c}
\bm{0}\!\! \\
\QC^{12}(k)\!\! \\
\end{array} \right]. \nonumber
\end{eqnarray}
Simplifying Eq.~\eqref{ext.crlb.12}, we get
\begin{eqnarray}
\lefteqn{\GQL(k+1) \approx\nonumber}\\
&&\QC^{22}(k)-\QC^{21}(k)\big(\GQL(k)+\QC^{11}(k)\big)^{-1}\QC^{12}(k),~~~~~~~~ \nonumber
\end{eqnarray}
where Eq.~\eqref{AppA:eqn4} has been used to obtain the final result.
This completes the proof for Result~2.
\end{proof}

\end{document}